\documentclass[runningheads,a4paper]{llncs}

\usepackage{booktabs}

\usepackage{latexsym}
\usepackage{mathrsfs}
\usepackage{amssymb}
\usepackage{algorithm}
\usepackage[noend]{algpseudocode}
\usepackage{pst-node,pst-plot,tikz}
\usepackage{tikz}
\usepackage{verbatim}

\usetikzlibrary{arrows,shapes,positioning,shadows,trees}

\tikzset{
  basic/.style  = {draw, text width=4cm, font=\sffamily, rectangle},
  root/.style   = {basic, rounded corners=2pt, thin, align=center},
  level 2/.style = {basic, rounded corners=2pt, thin, align=center},
  level 3/.style = {basic, rounded corners=2pt, thin, align=center},
   level 4/.style = {basic, rounded corners=2pt, thin, align=center}

}

\makeatletter
\def\BState{\State\hskip-\ALG@thistlm}

\title{On Learning a Hidden Directed Graph with Path Queries} 

\author{Mano Vikash Janardhanan\inst{1} \and Lev Reyzin\inst{2}}

\institute{Lifion by ADP \\ New York, NY 10011\\ 
\texttt{mano@manovikash.com}\\ $ $
\and
Department of Mathematics, Statistics, and Computer Science\\ University of Illinois at Chicago \\  Chicago IL 60607 \\
\texttt{lreyzin@uic.edu}}

\begin{document}

\maketitle

\begin{abstract}
In this paper, we consider the problem of reconstructing a directed graph using 
path queries.
In the query model of learning, a graph is hidden from the learner, and the learner can access
information about it with path queries.  For a source and destination node, 
a path query returns whether there is a directed path
from the source to the destination node in the hidden graph.  We first give bounds
for learning graphs on $n$ vertices and $k$ strongly connected components.  We then study
the case of bounded degree directed trees and give new algorithms for learning ``almost-trees'' -- directed trees to which
extra edges have been added.  We also give some lower bound constructions justifying our approach.

\keywords{Active learning  \and Graph algorithms \and Graph learning \and Path queries.}
\end{abstract}

\section{Introduction}

Problems in the area of \textbf{query learning of graphs} capture many different contexts.
In evolutionary tree reconstruction, an experimenter may measure or query the genetic
distance between two species with the goal of placing all the species onto one tree~\cite{Hein89,ReyzinS07b}.
In chemical reaction networks, one may view various chemicals as nodes in a hidden graph,
with edges corresponding to reacting pairs -- here an experimenter may mix chemicals to test for 
a reaction, which corresponds to querying subsets of vertices for the presence of an
independent set~\cite{AbasiBM18,AngluinC04,AngluinC06}.  Each real-world setting entails its own query learning model, in which
the learner typically tries to
reconstruct the (possibly weighted) adjacency information of the graph 
by making as few queries as possible, see e.g.~\cite{BouvelGK05,Reyzin09,ReyzinS07a,abasiB19,ChoiK08,Mazzawi10}. 

The model we study in this paper was introduced by Wang and Honorio~\cite{WangH19} 
and is a directed variant of other
well-studied models \cite{BeerliovaEEHHMR06,Hein89,JagadishS13}; it 
involves learning a directed graph by 
querying ordered pairs of vertices, testing for the presence of a directed path from
the first vertex in the pair to the second.
This model is meant to capture causality, answering the question ``when node $u$ is acted upon, does it create
a change in node $v$?'', but also has other applications, like
trying to learn the topology of the internet using ping requests from one IP address to another.


In particular, the model we consider herein is the following: the hidden target is a directed, unweighted graph,
and the queries are called \textbf{path queries}.  A path query consists of an ordered
pair of vertices $(u,v)$ and the result of the query is $1$ (or ``yes'') if the hidden graph has a directed
path from $u$ to $v$ and $0$ (or ``no'') otherwise.


In their work, Wang and Honorio~\cite{WangH19} prove the following:
Given a directed rooted tree with $n$ nodes and maximum degree at most $d$, there is a randomized algorithm which reconstructs the tree with expected query complexity $O(dn \log^2 n)$.
Their algorithm is recursive -- it picks two vertices at random, and with high probability, the path between those two vertices contains an edge which roughly splits the graph. 
Wang and Honorio show that finding the path and the edge has low query complexity. Then, they split the graph along this edge and recursively apply the technique to each part.

They also show 
an information theoretic lower bound of $\Omega (n\log n)$ and a lower bound of
 $\Omega (nd)$ (using a parallel chain construction) on the number of queries any algorithm must make. For general graphs, they show that
in order to reconstruct a sparse disconnected directed acyclic graph on $n$ nodes, any deterministic algorithm requires at least $\Omega(n^2)$ queries.
This proof involves differentiating between an empty graph and a single edge.
Finally, they show that
in order to reconstruct a sparse connected directed acyclic graph on $n$ nodes, any deterministic algorithm requires at least $\Omega(n^2)$ queries.

In this work, we extend the understanding of path queries by first considering the problem of
learning strongly connected components, as well as the edges between them (see Section~\ref{subsec:scc}). Then,
in our main contribution,
 in Section~\ref{subsec:rdg} we extend the results
of Wang and Honorio by tackling the problem of learning almost-trees (see Definition~\ref{def:at}).  Almost-trees
are trees with an extra edge.  In the case of evolutionary
trees, this begins to tackle real-world problems caused by processes like hybridization~\cite{Barton01}, where on occasion a species can have two distinct paths to an 
ancestor, breaking the expected tree-structure of evolution.
Our approach matches the bound of Wang and Honorio's algorithm (up to polylog factors) and is more general.

\section{Preliminaries}

Let $G=(V,E)$ be a directed graph with vertex set $V$ and edge set $E$. Let $(i,j)$ denote the directed edge from $i$ to $j$. We assume $|V|=n$. Two vertices $i,j$ are said to be strongly connected if there is a directed path from $i$ to $j$ and $j$ to $i$. This binary relation is an equivalence relation and the induced equivalence classes are called strongly connected components. Let $G$ have $k$ strongly connected components, the collection of strongly connected components $\{S_1,S_2,\ldots,S_k\}$ forms a partition of $V$.

A directed graph is called acyclic if it has no cycle. Hence, a directed graph is acyclic if and only if it has no strongly connected subgraphs with more than one vertex. 

If we start with an undirected graph $G$, pick a $r\in V$ called root and orient the edges such that there is a path from $r$ to all other $v\in V$, the resulting directed graph is called a rooted directed graph. If the undirected graph we started with was a tree, the resulting directed graph is called a rooted directed tree.

In their work, Wang and Honorio \cite{WangH19} define path queries as follows:

\begin{definition}[path query]
Let $G=(V,E)$ be a directed graph. A path query is a function $Q_G: V \times V \rightarrow \{0,1\}$ such that $Q_G(i,j)=1$ if there exists a path in $G$ from $i$ to $j$, and $Q_G(i,j)=0$ otherwise.
\end{definition}

They give an algorithm for reconstructing bounded-degree directed rooted trees and make observations on what type of edges are not learnable. In particular, they observe transitive edges are not learnable where transitive edges are defined as follows:

\begin{definition}[transitive edges]
Let $G=(V,E)$ be a directed graph. We say an edge $(i,j)\in E$ is transitive if there exists a directed path from $i$ to $j$ of length greater than 1.
\end{definition}
We give new algorithms for reconstructing bounded-degree directed graphs using path queries which work for regimes other than bounded degree directed rooted trees. Because it is not possible to learn transitive edges, we will either redefine the notion of learning when transitive edges are present in the graph in Section~\ref{subsec:scc} or consider promise instances where such edges are not present in Section~\ref{subsec:rdg}.

We now provide a few useful definitions that are needed for later.  We begin with notions of a layered graph and graph height, a useful definition of almost-trees, and the notions of descendants, of ancestors, and of a parent in a tree.

\begin{definition}[layered graph, graph height]
Given a rooted directed graph $G$ with root $r$, any tree $T\subseteq G$ which contains paths from $r$ to all other $v\in V$ is called a layered graph of $G$. The length of the longest path in $G$ from $r$ to any other $v\in V$ is denoted by $h$ and is called the height of $G$.
\end{definition}

\begin{definition}[almost-tree]\label{def:at}
A rooted directed graph $G$ with root $r$ is an almost-tree if $G$ is the union of a rooted directed tree and a single additional directed edge.
\end{definition}

\begin{definition}[descendants, ancestors, parent]
We define the descendant set, ancestor set and parent of a vertex $i$ as follows:
\begin{itemize}
\item $D(i) = \{ u: Q_G(i,u) = 1\}$
\item $A(i) = \{ u : Q_G(u,i) = 1\}$
\item For a rooted directed tree, let $p(i)$ denote the vertex which is the parent of $i$.
\end{itemize}
\end{definition}

Note that we can find both $D(i)$ and $A(i)$ with $2(n-1)$ queries by $Q_G(u,i)$ and $Q_G(i,u)$ for all $u\in V$.

\section{Learnability results}

We begin with some simpler results, which clarify the query complexity of recovering the strongly connected components of a graph.

\subsection{Strongly connected components}\label{subsec:scc}

Suppose $G$ has $k$ strongly connected components, then we have the following upper bound. Note that when we have strongly connected components, there are transitive edges and hence we cannot reconstruct all the edges within each component. Also, note that there could be transitive edges across components. For example, suppose there are three vertices $a$, $b$ and $c$ which are strongly connected components individually. Suppose there is an edge from $a$ to $b$ and another edge from $b$ to $c$, then the edge from $a$ to $c$ is transitive and cannot be learnt. Hence, assuming that there are no transitive edges across strongly connected components, the notion of learning here is to find the strongly connected components $\{S_1,S_2,\ldots,S_k\}$ of $G$ and for each $i,j\in [k]$, $i\neq j$, whether there is an edge between some vertex in $S_i$ to some vertex in $S_j$.

\begin{theorem}
Assuming that there are no transitive edges across strongly connected components, query complexity to learn a graph is $O(nk)$.
\end{theorem}
\begin{proof}
It follows from Proposition~$2$ of the work of Reyzin and Srivastava~\cite{ReyzinS07a} that we can recover the partition $\{S_1,S_2,\ldots,S_k\}$ in $O(nk)$ queries. 
Then, we can 
perform $O(k^2)$ queries to learn edges that go between two strongly connected components by querying any pair of vertices from each pair of the learned
strongly connected components.  Finally, we observe that $$O(nk)+O(k^2) = O(nk)$$ since it must be that $k \le n$.
\end{proof}


\subsection{Rooted directed graphs}\label{subsec:rdg}

For rooted directed graphs let us fix the notion of learning to completely reconstruct all the edges. This will be our definition of learnability for the rest of the paper. Note that for almost-trees, the additional edge should follow some natural properties for the problem to be well defined. Firstly, the extra edge cannot be a transitive edge. Also, if the additional edge goes from a node to an ancestor, a strongly connected component is created and it becomes impossible to reconstruct the edges in that component. An almost tree is defined to be path query reconstructable if all the edges can be recovered by path queries.

\subsection*{Lower bound}
We start with a lower bound. We give a lower bound of $\Omega(n^2)$ for path query reconstructable almost-trees with maximum degree $d=O(1)$.

\begin{theorem}\label{thm:lowerbound}
There exists a path query reconstructable almost-tree $G$ on $n-1$ vertices with maximum degree $d= O(1)$ such that any randomized algorithm to reconstruct $G$ requires at least $\Omega(n^2)$ queries in expectation.
\end{theorem}

\begin{proof}
Let us start by proving the result for a deterministic algorithm. Let $n$ be an even number and consider a caterpillar graph on $n-1$ vertices as shown in figure \ref{fig:caterpillar}. Assume $v_1$ is the root. Pick $i,j \in \{ 1,2, \ldots , n/2-1 \}$ uniformly at random such that $i<j$ and add the edge from $v_{n/2+i}$ to $v_{n/2+j}$. Even if the algorithm knows the caterpillar graph, it still needs to make $\Omega(n^2)$ queries to detect the random edge because presence of the edge only changes the single query $(v_i,v_j)$.

Now, let us apply Yao's minimax principle. Consider a uniform distribution over all random edges with $i<j$. For any fixed deterministic algorithm, the expected query complexity is $\Omega(n^2)$. Hence, for any randomized algorithm, there exists a worst case input such that the expected query complexity is $\Omega(n^2)$.
\begin{figure}[h!]

\[\begin{tikzpicture}[scale=1]
\node[fill=black, circle, inner sep=2pt] (n1) at (0,0) {};
\node[fill=black, circle, inner sep=2pt] (n2) at (1,1) {};
\node[fill=black, circle, inner sep=2pt] (n3) at (2,2) {};

\node[fill=black, circle, inner sep=2pt] (n3) at (4,4) {};
\node[fill=black, circle, inner sep=2pt] (n3) at (5,5) {};
\node[fill=black, circle, inner sep=2pt] (n3) at (6,6) {};

\node[fill=black, circle, inner sep=2pt] (n1) at (1.5,0) {};
\node[fill=black, circle, inner sep=2pt] (n2) at (2.5,1) {};
\node[fill=black, circle, inner sep=2pt] (n3) at (4.5,3) {};
\node[fill=black, circle, inner sep=2pt] (n3) at (5.5,4) {};
\node[fill=black, circle, inner sep=2pt] (n3) at (6.5,5) {};

\node (n1) at (-1,0)  {\small $v_{n/2}$};
\node (n2) at (0,1)  {\small $v_{n/2}-1$};
\node (n3) at (1,2)  {\small $v_{n/2}-2$};
\node (n1) at (5.5,6)  {\small $v_{1}$};
\node (n2) at (4.5,5)  {\small $v_2$};
\node (n3) at (3.5,4)  {\small $v_3$};

\node (n1) at (7.5,5)  {\small $v_{n/2+1}$};
\node (n2) at (6.5,4)  {\small $v_{n/2+2}$};
\node (n3) at (5.5,3)  {\small $v_{n/2+3}$};
\node (n1) at (4.,1)  {\small $v_{n/2+n/2-2}$};
\node (n2) at (2.5,0)  {\small $v_{n/2+n/2-1}$};

\draw[thick,loosely dotted] (2,2) to (4,4);
\draw[thick,loosely dotted] (2.5,1) to (4.5,3);

\draw(0,0) to (2,2);
\draw(4,4) to (6,6);

\draw(1,1) to (1.5,0);
\draw(2,2) to (2.5,1);

\draw(4,4) to (4.5,3);
\draw(5,5) to (5.5,4);
\draw(6,6) to (6.5,5);
\end{tikzpicture}\]
\caption{Caterpillar Graph\label{fig:caterpillar}}
\end{figure}
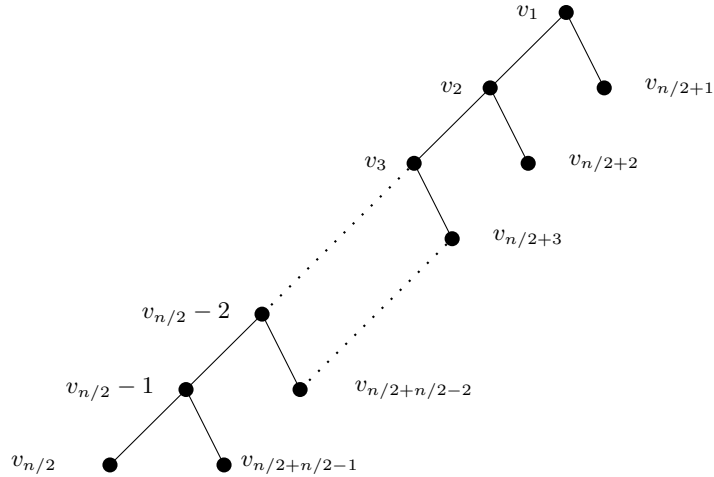

\end{proof}
Note that for a caterpillar graph, illustrated in Figure~\ref{fig:caterpillar}, $h=O(n)$.\footnote{We note that the same graph and Figure were employed in Janardhanan's work~\cite{Janardhanan17} on betweeness queries.} In Theorem \ref{thm:lowerbound_extended}, we extend this idea to get a lower bound as a function of $n$ and $h$.

\subsection*{Upper bound}

Our main result is an upper bound on the query complexity of Algorithm \ref{rootedgraph} which is a clean recursive randomized algorithm for learning an almost-tree. This algorithm can also be used to learn trees and hence generalises the main result in \cite{WangH19} with the loss of only an extra $O(\log n)$ factor. The upper bound on the query complexity of Algorithm \ref{rootedgraph} stated below asymptotically matches the lower bound in Theorem \ref{thm:lowerbound_extended} as a function of $n$ and $h$ ignoring the $\log$ factors.
 
The time complexities of the various subroutines of the algorithm are shown in Figure \ref{fig:algorithms}.
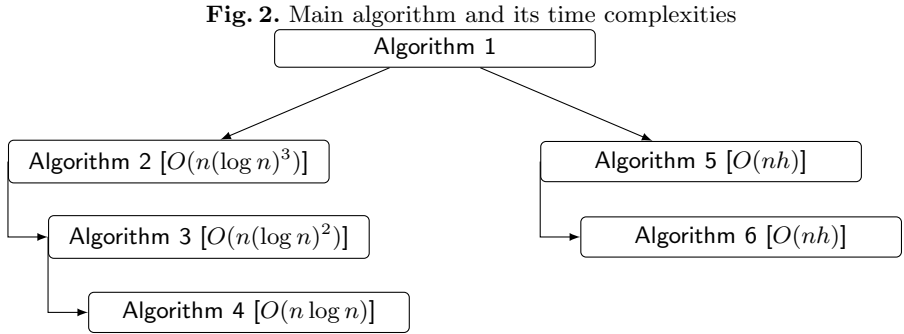
\begin{theorem}

\begin{figure}
\caption{Main algorithm and its time complexities} \label{fig:algorithms}
\begin{tikzpicture}[
  level 1/.style={sibling distance=70mm},
  edge from parent/.style={->,draw},
  >=latex]

\node[root] {Algorithm 1}
  child {node[level 2] (c1) {Algorithm 2 [$O(n (\log n)^3)$]}}
  child {node[level 2] (c2) {Algorithm 5 [$O(n h)$]}};

\begin{scope}[every node/.style={level 3}]
\node [below of = c1, xshift=15pt] (c11) {\mbox{Algorithm 3 [$O(n (\log n)^2)$]}};
\node [below of = c11, xshift=15pt] (c12) {Algorithm 4 [$O(n \log n)$]};
\node [below of = c2, xshift=15pt] (c21) {\mbox{Algorithm 6} [$O(nh)$]};
\end{scope}

\foreach \value in {1}
  \draw[->] (c1.west) |- (c1\value.west);

\foreach \value in {1}
  \draw[->] (c2.west) |- (c2\value.west);
\draw[->] (c11.west) |- (c12.west);
\end{tikzpicture}
\end{figure}

Algorithm \ref{rootedgraph} is a randomized algorithm that learns a path reconstructable almost-tree $G$ with maximum degree $d=O(1)$ using ${O}(n(\log n)^3+nh)$ path queries where $h$ is the height of $G$.
\end{theorem}

The idea behind Algorithm \ref{rootedgraph} is to first find a layered graph in $G$ (this is done in line \ref{line:layeredgraph_line} of the algorithm). As a layered graph (say $G_L$) is a tree, we get $n-1$ edges of $G$. This means that we have reconstructed a spanning tree of $G$ and we are left with the task of finding one more edge in $G$ as we know that $G$ is an almost-tree. Let us call this edge a cross edge. In other words, a cross edge is the edge in $G$ that is not in the layered graph produced by line \ref{line:layeredgraph_line} of Algorithm \ref{rootedgraph}. The next task is to find the cross edge. This is done in line \ref{line:crossedges_line} of Algorithm \ref{rootedgraph}. 
To find the layered graph structure, in line \ref{line:layeredgraph_line} of Algorithm \ref{rootedgraph}, we call Algorithm~\ref{rootedgraph_layers}. This algorithm works recursively by finding an edge whose descendant set roughly splits the graph into equal parts. Hence, the depth of the recursion tree is $O(\log n)$. We show that with high probability, the randomized algorithm which finds such an edge (Algorithm \ref{splitgraph}) on a subset of vertices $V$ uses ${O}(|V| (\log |V|)^2)$ queries. This gives an overall query complexity of $O(n(\log n)^3)$ for finding a layered graph in $G$.

We need the following structure theorem rephrased from 
\cite{AbrahamsenBRS16}.
\begin{lemma}
Let $G=(V,E)$ be a directed rooted graph with root $r$ and maximum degree $d$. For any $v\in V$, there exists a $w\in D(v)$ such that
$$ \left\lceil \frac{|D(v)|}{3d}\right\rceil\leq |D(w)| \leq \left\lceil\frac{|D(v)|}{3}\right\rceil $$
\end{lemma}
We call $w$ which roughly splits $D(v)$ as a splittable vertex.

\begin{algorithm}
\caption{Reconstruct Rooted Graph}\label{rootedgraph}
\begin{algorithmic}[1]
\Function{R{\scriptsize ECONSTRUCT}-R{\scriptsize OOTED}-G{\scriptsize RAPH}($V$)}{} 
\State $G_L=$ R{\scriptsize ECONSTRUCT}-L{\scriptsize AYERED}-G{\scriptsize RAPH}($V$)  
\algorithmiccomment{finds edges in layered graph}\label{line:layeredgraph_line}
\State $G_C=$ F{\scriptsize IND}-C{\scriptsize ROSS}-E{\scriptsize DGES}($G_L$) \label{line:crossedges_line}
\algorithmiccomment{finds cross-edges}
\State \Return $G_L\cup G_C $
\EndFunction
\end{algorithmic}
\end{algorithm}

\begin{algorithm}
\caption{Reconstruct Layered Graph}\label{rootedgraph_layers}
\begin{algorithmic}[1]
\Function{R{\scriptsize ECONSTRUCT}-L{\scriptsize AYERED}-G{\scriptsize RAPH}($V$)}{}
\If{$|V|\leq 2$}
\State \Return $V_G$ \algorithmiccomment{$V_G$ is the subgraph induced on $V$ by $G$}
\EndIf
\State $(V_1,V_2, e) \gets $S{\scriptsize PLIT}-G{\scriptsize RAPH}($V$)
\algorithmiccomment{finds $e$ that splits the graph into $V_1$ and $V_2$}
\State $G_1=$ R{\scriptsize ECONSTRUCT}-L{\scriptsize AYERED}-G{\scriptsize RAPH}($V_1$)
\algorithmiccomment{recurse in $V_1$}
\State $G_2=$ R{\scriptsize ECONSTRUCT}-L{\scriptsize AYERED}-G{\scriptsize RAPH}($V_2$)
\algorithmiccomment{recurse in $V_2$}
\State \Return $G_1\cup G_2 \cup \{e\}$
\EndFunction
\end{algorithmic}
\end{algorithm}

We start by analyzing the search sub-routine (Algorithm~\ref{search}). We claim that if the random vertex picked in line \ref{line:rand1} of the subroutine is in $D(s)$ for some splittable vertex $s$, then the subroutine will find $s$.  This is because if the random vertex picked in line \ref{line:rand1} is in $D(s)$ for some splittable vertex $s$, then we enter the $|D(i)| < |V|/3d$ case corresponding to the if statement in line \ref{line:iindes}. Inside this if statement, we set $P$ to be the set of potential splittable vertices among the ancestors of $i$ and keep updating $P$ by doing randomized binary search for $s$. This proves the claim and hence, the number of times the search sub-routine is called is $O(\log n)$.

\begin{remark}
Note that the presence of the additional edge in the almost-tree does not affect the overall structure of the algorithm. Its presence will be felt only in the randomised binary search (in Algorithm~\ref{search}) as there may be two paths from the root to vertex $i$. If both paths contain a splittable vertex, then the randomised binary search will find it. The difficult case is when only one of the two paths contains a splittable vertex. Suppose there are two paths $P$ containing the splittable vertex and $P'$ not containing the splittable vertex. Every time the randomised binary search picks a vertex in $i\in P' \setminus P$, we will end up deleting all the children of $i$ from the potential vertices for the next iteration of the randomised binary search in line \ref{line:deletevertices} of Algorithm~\ref{search}.
\end{remark}

\begin{algorithm}
\caption{Split Graph}\label{splitgraph}
\begin{algorithmic}[1]
\Function{S{\scriptsize PLIT}-G{\scriptsize RAPH}($V$)}{}
\State $splittable = 0$
\algorithmiccomment{keeps track of whether the output of S{\scriptsize EARCH}($V$) is splittable}
\While{$splittable = 0$}
\algorithmiccomment{repeat until the output of S{\scriptsize EARCH}($V$) is splittable}
\State $(splittable, v) =$ S{\scriptsize EARCH}($V$)
\EndWhile

\State \Return $(D(v),V\setminus D(v), (p(v),v))$
\EndFunction
\end{algorithmic}
\end{algorithm}

To analyze the query complexity of Algorithm~\ref{search}, the search sub-routine, we first note that the queries are only made in line \ref{line:query_review} and \ref{line:rand2}. In each call of search, line \ref{line:rand1} gets executed once and hence makes $O(n)$ queries. Line \ref{line:rand2} gets executed $O(\log n)$ times in expectation because it is a randomized binary search and hence makes $O(n\log n)$ queries in expectation. So, the expected query complexity of the search sub-routine is $O(n\log n)$.

The expected query complexity of Algorithm \ref{splitgraph} (split graph) is $O(n(\log n)^2)$ and the expected query complexity of Algorithm \ref{rootedgraph_layers} (reconstruct layered graph) is $O(n(\log n)^3)$.

\begin{algorithm}
\caption{Search}\label{search}
\begin{algorithmic}[1]
\Function{S{\scriptsize EARCH}($V$)}{}
\State $P=V$
\State pick $i\in P$ randomly  \label{line:rand1}
\algorithmiccomment{pick a random vertex $i$}
\State $\forall\ u \in V$, query $(u,i)$ and $(i,u)$  \label{line:query_review}
\If{$|D(i)|>|V|/3$}
\algorithmiccomment{check if $i$ and $i$'s ancestors are both not splittable}
\State \Return $(0,0)$
\EndIf
\If{$|V|/3d\leq D(i)\leq  |V|/3$}
\algorithmiccomment{check if $i$ is splittable}
\State \Return $(1,i)$
\EndIf
\If{$|D(i)| < |V|/3d$} \label{line:iindes}
\algorithmiccomment{check if $i$'s ancestors may be splittable}
\State $P=A(i)$
\While{$P\neq\emptyset$}
\algorithmiccomment{search over $i$'s ancestors for splittable vertices}
\State pick $i \in P$ randomly and query \\ 
\ \ \ \ \ \ \ \ \ \ \ \ \ \ $(u,i)$ and $(i,u)\ \forall u\in V$. \label{line:rand2}
\If{$|V|/3d\leq |D(i)|\leq |V|/3$}
\State \Return $(1,i)$
\EndIf
\If{$|D(i)|>|V|/3$}
\State $P = (P \cap D(i))\setminus\{i\}$
\EndIf
\If{$|D(i)| < |V|/3d$} 
\State $P = P \setminus D(i)$ \label{line:deletevertices}
\EndIf
\EndWhile  
\State \Return (0,0)

\EndIf

\EndFunction
\end{algorithmic}
\end{algorithm}

Now we analyze the algorithm to find cross edges (Algorithm \ref{cross_edges}). We assume that all sub-routines under this algorithm have access to $G_L$. Algorithm \ref{cross_edges} first calls Algorithm \ref{cross_edges_recursive}, 
which is a recursive procedure for finding a triplet of vertices $v,a,b$ that satisfy the following conditions:
\begin{itemize}
\item $v=p(a)$.
\item $b$ is a leaf vertex.
\item $a$ and $b$ belong to different subtrees of $v$ under $G_L$.
\item $Q(a,b)=1$. 
\end{itemize}

We refer to such a vertex $v$ as a top vertex.
This must mean that the extra edge has caused $Q(a,b)=1$. Once we know $v,a,b$, we can find the extra edge exactly by traversing $G_L$. This is done in the F{\scriptsize IND}-C{\scriptsize ROSS}-E{\scriptsize DGES}-S{\scriptsize PECIFIC} algorithm.

Now, we turn to analyzing Algorithm \ref{cross_edges_recursive}. We start from the root $r$. If $r$ is a top vertex, then the queries in line \ref{line:query1} of Algorithm \ref{cross_edges_recursive} will find it. If not, we have split the problem into smaller subproblems corresponding to the descendant set of each immediate child of $r$. We recursively call the same function for each immediate child in line \ref{line:recurse1} of Algorithm \ref{cross_edges_recursive}.

F{\scriptsize IND}-C{\scriptsize ROSS}-E{\scriptsize DGES}-S{\scriptsize PECIFIC} algorithm is given $v,a,b$ as input where $v$ is a top vertex. Hence, the cross edge is of the form $(c_1,c_2)$ where $c_1\in D(a)$ and $c_2\in A(b)$. Therefore, this algorithm simply traverses over the descendants of $a$ starting from the immediate children of $a$ and the ancestors of $b$ starting from the immediate parent of $b$ until it finds $(c_1,c_2)$ exactly.

\begin{algorithm}
\caption{Find Cross Edges}\label{cross_edges}
\begin{algorithmic}[1]
\Function{F{\scriptsize IND}-C{\scriptsize ROSS}-E{\scriptsize DGES}($G_L$)}{}
\State $(found, v, a, b)=$ F{\scriptsize IND}-C{\scriptsize ROSS}-E{\scriptsize DGES}-R{\scriptsize ECURSIVE}($r$)
\State $(p,q)=$ F{\scriptsize IND}-C{\scriptsize ROSS}-E{\scriptsize DGES}-S{\scriptsize PECIFIC}($v,a,b$)
\State \Return $(p,q) $
\EndFunction
\end{algorithmic}
\end{algorithm}

The number of queries made in line \ref{line:query1} of Algorithm \ref{cross_edges_recursive} is $O(n)$ assuming that $d=O(1)$. Suppose the height of $G_L$ is $h$, the depth of the recursion is also $h$. Hence, the query complexity is $O(nh)$.

\begin{algorithm}
\caption{Find Cross Edges Recursive}\label{cross_edges_recursive}
\begin{algorithmic}[1]
\Function{F{\scriptsize IND}-C{\scriptsize ROSS}-E{\scriptsize DGES}-R{\scriptsize ECURSIVE}($v$)}{}
\State Let $C$ be the immediate children of $v$ in $G_L$.
\For{$c \in C$}
\State Let $P_c= \cup_{k\in C, k\neq c}L_k$ where $L_k$ are the leaves in $D(k)$.
\State $\forall l\in P_c$, query $(c,l)$ \label{line:query1}
\If{$Q(c,l)=1$ for some query in line \ref{line:query1}}
\State \Return $(1,v,c,l)$
\EndIf
\EndFor

\For{$c \in C$}
\State $(found, k, a, b) = $ F{\scriptsize IND}-C{\scriptsize ROSS}-E{\scriptsize DGES}-R{\scriptsize ECURSIVE}($c$) \label{line:recurse1}
\If{$found = 1$}
\State \Return $(1,k, a, b)$
\EndIf
\EndFor

\EndFunction
\end{algorithmic}
\end{algorithm}

Hence, the overall query complexity of Algorithm \ref{rootedgraph} is $O(n(\log n)^3 +nh)$. Note that when $G_L$ is a complete $d$-ary trees, this gives a $\tilde{O}(n)$ algorithm where as for caterpillar graphs, we get a $O(n^2)$ algorithm which matches with the lower bound proved in Theorem \ref{thm:lowerbound}. We can also extend Theorem \ref{thm:lowerbound} for arbitrary $h$ to get a result of the following form:

\begin{theorem}\label{thm:lowerbound_extended}
For every $d = O(1)$ and $h> (1+ c) \log_d n$ for $c>0$, there exists a path query reconstructable almost-tree $G$ on $n-1$ vertices with maximum degree $d$ and height $h$ such that any deterministic algorithm to reconstruct $G$ requires at least $\Omega(nh)$ queries.
\end{theorem}

The proof involves modifying Theorem \ref{thm:lowerbound} by considering a caterpillar graph on $\Theta(h)$ vertices and a complete $d$-ary tree with $\Theta(n)$ leaves attached to the last level of the caterpillar graph. This construction can be made to work for any $h> (1+ c) \log_d n$ for $c>0$. Now, add an edge randomly from one of the leaves of the caterpillar graph to one of the leaves of the complete $d$-ary tree. Detecting this random edge requires $\Omega(nh)$ queries as there are $\Theta(h)$ leaves in the caterpillar graph and $\Theta(n)$ leaves in the complete $d$-ary tree.

\section{Open problems}
Our algorithm works when there is exactly one extra edge. It would be interesting to see if this approach can be generalised when multiple extra edges are present. It would also be interesting to find an algorithm where the query complexity is a function of the number of extra edges. 

\section*{Acknowledgements}
This work by supported in part by NSF grants CCF-1848966 and CCF-1934915.

\bibliography{paper}

\end{document}